\RequirePackage{amsmath}
\documentclass{llncs}
\usepackage{makeidx}  
\usepackage{amssymb,authblk}
\usepackage{t1enc}
\usepackage[utf8]{inputenc}
\usepackage{float}
\usepackage{tikz}
\usepackage{cite}
\usepackage{amsrefs}
\usetikzlibrary{calc}
\usetikzlibrary{decorations}
\usetikzlibrary{decorations.pathreplacing}
\usepackage{tabularx}
\newcolumntype{C}{>{\centering\arraybackslash}m{28pt}}

\begin{document}

\title{On the complexity of color-avoiding site and bond percolation}
\titlerunning{On the complexity of color-avoiding percolation}  
%
\author{Roland Molontay \inst{1} \and Kitti Varga \inst{2}}
\authorrunning{R. Molontay, K. Varga} 
%
%
\institute{Department of Stochastics, Budapest University of Technology and Economics, Budapest, Hungary\\
MTA-BME Stochastics Research Group, Budapest, Hungary \\
\email{molontay@math.bme.hu}
\and
Department of Computer Science and Information Theory, Budapest University of Technology and Economics, Budapest, Hungary\\
\email{vkitti@cs.bme.hu}}

\maketitle              

\begin{abstract}
The mathematical analysis of robustness and error-tolerance of complex networks has been in the center of research interest. On the other hand, little work has been done when the attack-tolerance of the vertices or edges are not independent but certain classes of vertices or edges share a mutual vulnerability. In this study, we consider a graph and we assign colors to the vertices or edges, where the color-classes correspond to the shared vulnerabilities. An important problem is to find robustly connected vertex sets: nodes that remain connected to each other by paths providing any type of error (i.e. erasing any vertices or edges of the given color). This is also known as color-avoiding percolation.

In this paper, we study various possible modeling approaches of shared vulnerabilities, we analyze the computational complexity of finding the robustly (color-avoiding) connected components. We find that the presented approaches differ significantly regarding their complexity.
\end{abstract}

\section{Introduction and related works}

Understanding the attack and error tolerance of complex networks -- i.e. the ability to maintain the overall connectivity of the network as the vertices (or edges) are removed -- has attracted a great deal of research interest in the last two decades \cites{albert2000error, callaway2000network, zhao2005tolerance, shao2015percolation, barabasi2016network}. Most of the works have focused on random error, meaning that the nodes are considered to be homogeneous with respect to their vulnerabilities, or hub-targeted attack, i.e. the nodes fail preferentially according to a structural property such as degree or betweenness centrality. However, real-world networks are typically heterogeneous (not only with respect to their degree distribution): nodes (or edges) can be separated into different classes regarding a mutually shared vulnerability within the class. The shared vulnerabilities can be modeled by assigning a color to each class that may represent a shared eavesdropper, a controlling entity or correlated failures. A "color-avoiding" percolation framework was developed by Krause et al. \cites{krause2016hidden, krause2017color, kadovic2018bond}.

Traditional percolation theory can be used to study the behaviour of connected components in a graph if a vertex (site) or edge (bond) failure occurs with a given probability \cites{stauffer2014introduction, newman2018networks}. In the traditional approach a single path provides connectivity, however, here connectivity corresponds to the ability to avoid all vulnerable sets of vertices or edges via multiple paths such that no color is required for all paths (color-avoiding connectivity).

This question is different from $k$-core percolation where any $k$ paths
are sufficient between two nodes \cites{yuan2016k, dorogovtsev2006k}; and also different from $k$-connectivity, where $k$ mutually independent paths are required \cite{penrose1999k}. Another related concept is percolation on multiplex networks \cites{son2012percolation, hackett2016bond} where the layers can be thought of as the colors, but this approach also differs in the definition of connectivity \cite{kadovic2018bond}.

To study color-avoiding percolation a new framework was needed. The problem was introduced by Krause et al. \cite{krause2016hidden} who analyzed the color-avoiding connectivity of networks with shared vulnerabilities on the vertices. The authors also examined the latent color-avoiding connectivity structure of the AS-level Internet \cite{siganos2003power} where the color of the node represents the country to which the router is registered. They have found that 26{,}228 out of 49{,}743 of the routers are in the largest color-avoiding component, i.e. secure communication can be obtained among them by splitting the message into more pieces and transmitting them on different paths - even if every country is eavesdropping on its traffic \cite{krause2016hidden}. In \cite{krause2017color} the theory of color-avoiding percolation has been extended. The authors study analytically and numerically the maximal set of nodes that are color-avoiding connected in random networks with randomly
distributed colors. Shekhtman et al. \cite{shekhtman2018critical} generalize this framework to study secure message-passing in networks with a given community structure and different classes of vulnerabilities. Kadovi\'c et al. \cite{kadovic2018bond} formulated color-avoiding percolation for colored edges as well and studied color-avoiding bond and site percolation for networks with a power-law degree distribution.

Beside color-avoiding percolation, finding paths in colored graphs has gained interest in other domains as well. Wu \cite{wu2012maximum} introduced the Maximum Colored Disjoint Paths (\textsc{MaxCDP}) problem that is to find the maximum number of vertex-disjoint paths with edges of the same color between two vertices. The complexity of \textsc{MaxCDP} has been investigated: it can be solved in polynomial time if the input graph contains one color but if there are at least two colors the problem is NP-hard \cites{wu2012maximum, gourves2012paths}. An even harder variant of the problem, to find the maximum number of vertex-disjoint and color-disjoint uni-color paths, was introduced by Dondi et al. \cite{dondi2016finding}.

Another related problem is finding $k$-multicolor paths in a graph. Santos et al. \cite{santos2017multicolour} show that the problem is NP-hard and it is also hard to approximate. Another interesting problem is finding a path from a vertex that meets all the colors in a graph. This problem is NP-hard considering a properly colored directed graph, as well as finding the shortest or longest such paths \cite{granata2012complexity}.  A survey on algorithmic and computational results for other coloring problems can be found in \cite{malaguti2010survey}.

On the other hand, none of the previously mentioned works have addressed the computational complexity of color-avoiding percolation. In this article we present various modeling approaches to handle shared vulnerabilities along with the analysis regarding computational complexity.  Section \ref{model} is devoted to the problem definition considering both vertices and edges as the targets of the attack. In Section \ref{complex} we conduct complexity analysis and find that, although the presented approaches are seemingly very similar, they differ significantly regarding computational complexity. Section \ref{conlusion} concludes the work.

\section{Modeling shared vulnerabilities in networks}
\label{model}


There are two main approaches of modeling shared vulnerabilities in networks depending on the target of the attack: either the links between the nodes can be destroyed (or eavesdropped) or the nodes themselves are the subject of a possible failure. The former is modeled by coloring the edges according to the shared vulnerabilities leading to color-avoiding bond percolation, while the latter is represented by assigning a color to the vertices having the same vulnerability resulting in color-avoiding site percolation. The color-avoiding edge- and  vertex-connectivity is illustrated in Fig.~\ref{fig:colored_graph} on an Erd\H{o}s-R\'enyi random graph, the exact concepts will be introduced later in this section.

\subsection{Coloring the edges}

It is natural to consider the case when the edges with shared vulnerabilities are the subject of the attack and the vertices remain indistinguishable and unharmed. Some possible real-world examples with edges having shared vulnerabilities: different means of transportation (bus, underground, railway etc.) for a traffic network,  various metabolic pathways depending on particular biochemical profiles \cite{kadovic2017color}.

\begin{figure}
 \centering
  \begin{tikzpicture}
  \tikzstyle{vertex}=[draw,circle,fill,minimum size=5,inner sep=0]
  \tikzstyle{vertex_red}=[draw,circle,fill=red,minimum size=7,inner sep=0]
  \tikzstyle{vertex_blue}=[draw,circle,fill=blue,minimum size=7,inner sep=0]
  \tikzstyle{vertex_green}=[draw,circle,fill=green,minimum size=7,inner sep=0]
  \tikzstyle{vertex_deleted}=[draw,circle,dashed,minimum size=7,inner sep=0]
  
  \tikzstyle{vertex2}=[draw,circle,fill,minimum size=4,inner sep=0]
  \tikzstyle{vertex2_red}=[draw,circle,fill=red,minimum size=6,inner sep=0]
  \tikzstyle{vertex2_blue}=[draw,circle,fill=blue,minimum size=6,inner sep=0]
  \tikzstyle{vertex2_green}=[draw,circle,fill=green,minimum size=6,inner sep=0]
  \tikzstyle{vertex2_deleted}=[draw,circle,dashed,minimum size=6,inner sep=0]
  
  \begin{scope}[shift={(-3,0)}, scale=0.9]
  \node[vertex] (v1) at (-1.75,0) [label={[xshift=2pt, yshift=-2pt] $1$}] {};
  \node[vertex] (v2) at (0.75,-1.25) [label=left: $2$] {};
  \node[vertex] (v3) at (2.25,0.25) [label=above: $3$] {};
  \node[vertex] (v4) at (-2.75,0.75) [label=above: $4$] {};
  \node[vertex] (v5) at (2,-0.75) [label=right: $5$] {};
  \node[vertex] (v6) at (0,0) [label=above: $6$] {};
  \node[vertex] (v7) at (0.5,-1.75) [label=below: $7$] {};
  \node[vertex] (v8) at (1.5,0.6) [label=above: $8$] {};
  \node[vertex] (v9) at (-1.25,-1.25) [label={[xshift=6pt, yshift=-7pt] $9$}] {};
  \node[vertex] (v10) at (1.5,-1.75) [label=below: $10$] {};
  \node[vertex] (v11) at (1.1,0.1) [label={[xshift=6pt, yshift=-7pt] $11$}] {};
  \node[vertex] (v12) at (-1,-0.5) [label={[xshift=7pt, yshift=-11pt] $12$}] {};
  \node[vertex] (v13) at (1,-0.75) [label={[xshift=6pt, yshift=-4pt] $13$}] {};
  \node[vertex] (v14) at (0.5,1) [label=above: $14$] {};
  \node[vertex] (v15) at (-1.5,0.8) [label=above: $15$] {};
  \node[vertex] (v16) at (-3,-1.5) [label=below: $16$] {};
  \node[vertex] (v17) at (-0.75,0.75) [label=above: $17$] {};
  \node[vertex] (v18) at (-2.5,-0.25) [label=left: $18$] {};
  \node[vertex] (v19) at (-2.5,-1) [label={[xshift=-7pt, yshift=-7pt] $19$}] {};
  \node[vertex] (v20) at (-0.75,-1.75) [label=below: $20$] {};
  
  \draw[very thick, blue] (v1) -- (v4);
  \draw[very thick, green] (v1) -- (v6);
  \draw[very thick, blue] (v1) -- (v9);
  \draw[very thick, red] (v1) -- (v17);
  \draw[very thick, blue] (v1) -- (v18);
  \draw[very thick, red] (v1) -- (v19);
  
  \draw[very thick, green] (v2) -- (v5);
  \draw[very thick, green] (v2) -- (v6);
  \draw[very thick, red] (v2) -- (v7);
  \draw[very thick, blue] (v2) -- (v13);
  
  \draw[very thick, red] (v3) -- (v5);
  \draw[very thick, red] (v3) -- (v8);
  
  \draw[very thick, green] (v4) -- (v18);
  
  \draw[very thick, blue] (v5) -- (v10);
  \draw[very thick, red] (v5) -- (v11);
  \draw[very thick, red] (v5) -- (v13);
  
  \draw[very thick, blue] (v6) -- (v8);
  \draw[very thick, red] (v6) -- (v11);
  \draw[very thick, red] (v6) -- (v12);
  \draw[very thick, green] (v6) -- (v13);
  \draw[very thick, blue] (v6) -- (v15);
  
  \draw[very thick, green] (v7) -- (v10);
  \draw[very thick, red] (v7) -- (v20);
  
  \draw[very thick, green] (v8) -- (v14);
  
  \draw[very thick, blue] (v9) -- (v12);
  \draw[very thick, red] (v9) -- (v20);
  
  \draw[very thick, red] (v11) -- (v14);
  
  \draw[very thick, red] (v12) -- (v15);
  \draw[very thick, red] (v12) -- (v17);
  \draw[very thick, green] (v12) -- (v19);
  
  \draw[very thick, red] (v14) -- (v17);
  
  \draw[very thick, blue] (v15) -- (v18);
  
  \draw[very thick, green] (v16) -- (v19);
  
  \draw[very thick, red] (v18) -- (v19);
  
  \draw[very thick, red] (v19) -- (v20);
  \end{scope}
  
  \begin{scope}[shift={(3,0)}, scale=0.9]
  \node[vertex_blue] (v1) at (-1.75,0) [label={[xshift=2pt, yshift=-2pt] $1$}] {};
  \node[vertex_blue] (v2) at (0.75,-1.25) [label=left: $2$] {};
  \node[vertex_green] (v3) at (2.25,0.25) [label=above: $3$] {};
  \node[vertex_red] (v4) at (-2.75,0.75) [label=above: $4$] {};
  \node[vertex_green] (v5) at (2,-0.75) [label=right: $5$] {};
  \node[vertex_green] (v6) at (0,0) [label=above: $6$] {};
  \node[vertex_red] (v7) at (0.5,-1.75) [label=below: $7$] {};
  \node[vertex_green] (v8) at (1.5,0.6) [label=above: $8$] {};
  \node[vertex_green] (v9) at (-1.25,-1.25) [label={[xshift=6pt, yshift=-7pt] $9$}] {};
  \node[vertex_green] (v10) at (1.5,-1.75) [label=below: $10$] {};
  \node[vertex_red] (v11) at (1.1,0.1) [label={[xshift=6pt, yshift=-7pt] $11$}] {};
  \node[vertex_red] (v12) at (-1,-0.5) [label={[xshift=7pt, yshift=-11pt] $12$}] {};
  \node[vertex_blue] (v13) at (1,-0.75) [label={[xshift=6pt, yshift=-4pt] $13$}] {};
  \node[vertex_red] (v14) at (0.5,1) [label=above: $14$] {};
  \node[vertex_red] (v15) at (-1.5,0.8) [label=above: $15$] {};
  \node[vertex_green] (v16) at (-3,-1.5) [label=below: $16$] {};
  \node[vertex_blue] (v17) at (-0.75,0.75) [label=above: $17$] {};
  \node[vertex_red] (v18) at (-2.5,-0.25) [label=left: $18$] {};
  \node[vertex_green] (v19) at (-2.5,-1) [label={[xshift=-7pt, yshift=-7pt] $19$}] {};
  \node[vertex_green] (v20) at (-0.75,-1.75) [label=below: $20$] {};
  
  \draw[thick] (v1) -- (v4);
  \draw[thick] (v1) -- (v6);
  \draw[thick] (v1) -- (v9);
  \draw[thick] (v1) -- (v17);
  \draw[thick] (v1) -- (v18);
  \draw[thick] (v1) -- (v19);
  
  \draw[thick] (v2) -- (v5);
  \draw[thick] (v2) -- (v6);
  \draw[thick] (v2) -- (v7);
  \draw[thick] (v2) -- (v13);
  
  \draw[thick] (v3) -- (v5);
  \draw[thick] (v3) -- (v8);
  
  \draw[thick] (v4) -- (v18);
  
  \draw[thick] (v5) -- (v10);
  \draw[thick] (v5) -- (v11);
  \draw[thick] (v5) -- (v13);
  
  \draw[thick] (v6) -- (v8);
  \draw[thick] (v6) -- (v11);
  \draw[thick] (v6) -- (v12);
  \draw[thick] (v6) -- (v13);
  \draw[thick] (v6) -- (v15);
  
  \draw[thick] (v7) -- (v10);
  \draw[thick] (v7) -- (v20);
  
  \draw[thick] (v8) -- (v14);
  
  \draw[thick] (v9) -- (v12);
  \draw[thick] (v9) -- (v20);
  
  \draw[thick] (v11) -- (v14);
  
  \draw[thick] (v12) -- (v15);
  \draw[thick] (v12) -- (v17);
  \draw[thick] (v12) -- (v19);
  
  \draw[thick] (v14) -- (v17);
  
  \draw[thick] (v15) -- (v18);
  
  \draw[thick] (v16) -- (v19);
  
  \draw[thick] (v18) -- (v19);
  
  \draw[thick] (v19) -- (v20);
  \end{scope}
  
  \begin{scope}[shift={(-4.125,-4)}, scale=0.6]
  \node[vertex2] (v1) at (-1.75,0) [label={[xshift=2pt, yshift=-2pt] \scriptsize{$1$}}] {};
  \node[vertex2] (v2) at (0.75,-1.25) [label=left: \scriptsize{$2$}] {};
  \node[vertex2] (v3) at (2.25,0.25) [label=above: \scriptsize{$3$}] {};
  \node[vertex2] (v4) at (-2.75,0.75) [label=above: \scriptsize{$4$}] {};
  \node[vertex2] (v5) at (2,-0.75) [label=right: \scriptsize{$5$}] {};
  \node[vertex2] (v6) at (0,0) [label=above: \scriptsize{$6$}] {};
  \node[vertex2] (v7) at (0.5,-1.75) [label=below: \scriptsize{$7$}] {};
  \node[vertex2] (v8) at (1.5,0.6) [label=above: \scriptsize{$8$}] {};
  \node[vertex2] (v9) at (-1.25,-1.25) [label={[xshift=6pt, yshift=-7pt] \scriptsize{$9$}}] {};
  \node[vertex2] (v10) at (1.5,-1.75) [label=below: \scriptsize{$10$}] {};
  \node[vertex2] (v11) at (1.1,0.1) [label={[xshift=6pt, yshift=-7pt] \scriptsize{$11$}}] {};
  \node[vertex2] (v12) at (-1,-0.5) [label={[xshift=7pt, yshift=-11pt] \scriptsize{$12$}}] {};
  \node[vertex2] (v13) at (1,-0.75) [label={[xshift=6pt, yshift=-4pt] \scriptsize{$13$}}] {};
  \node[vertex2] (v14) at (0.5,1) [label=above: \scriptsize{$14$}] {};
  \node[vertex2] (v15) at (-1.5,0.8) [label=above: \scriptsize{$15$}] {};
  \node[vertex2] (v16) at (-3,-1.5) [label=below: \scriptsize{$16$}] {};
  \node[vertex2] (v17) at (-0.75,0.75) [label=above: \scriptsize{$17$}] {};
  \node[vertex2] (v18) at (-2.5,-0.25) [label=left: \scriptsize{$18$}] {};
  \node[vertex2] (v19) at (-2.5,-1) [label={[xshift=-7pt, yshift=-7pt] \scriptsize{$19$}}] {};
  \node[vertex2] (v20) at (-0.75,-1.75) [label=below: \scriptsize{$20$}] {};
  
  \draw[thick, blue] (v1) -- (v4);
  \draw[thick, blue] (v1) -- (v9);
  \draw[thick, red] (v1) -- (v17);
  \draw[thick, blue] (v1) -- (v18);
  \draw[thick, red] (v1) -- (v19);

  \draw[thick, red] (v2) -- (v7);
  \draw[thick, blue] (v2) -- (v13);
  
  \draw[thick, red] (v3) -- (v5);
  \draw[thick, red] (v3) -- (v8);
  
  \draw[thick, blue] (v5) -- (v10);
  \draw[thick, red] (v5) -- (v11);
  \draw[thick, red] (v5) -- (v13);
  
  \draw[thick, blue] (v6) -- (v8);
  \draw[thick, red] (v6) -- (v11);
  \draw[thick, red] (v6) -- (v12);
  \draw[thick, blue] (v6) -- (v15);
  
  \draw[thick, red] (v7) -- (v20);
  
  \draw[thick, blue] (v9) -- (v12);
  \draw[thick, red] (v9) -- (v20);
  
  \draw[thick, red] (v11) -- (v14);
  
  \draw[thick, red] (v12) -- (v15);
  \draw[thick, red] (v12) -- (v17);
  
  \draw[thick, red] (v14) -- (v17);
  
  \draw[thick, blue] (v15) -- (v18);

  \draw[thick, red] (v18) -- (v19);
  
  \draw[thick, red] (v19) -- (v20);
  \end{scope}
  
  \begin{scope}[shift={(0,-4)}, scale=0.6]
  \node[vertex2] (v1) at (-1.75,0) [label={[xshift=2pt, yshift=-2pt] \scriptsize{$1$}}] {};
  \node[vertex2] (v2) at (0.75,-1.25) [label=left: \scriptsize{$2$}] {};
  \node[vertex2] (v3) at (2.25,0.25) [label=above: \scriptsize{$3$}] {};
  \node[vertex2] (v4) at (-2.75,0.75) [label=above: \scriptsize{$4$}] {};
  \node[vertex2] (v5) at (2,-0.75) [label=right: \scriptsize{$5$}] {};
  \node[vertex2] (v6) at (0,0) [label=above: \scriptsize{$6$}] {};
  \node[vertex2] (v7) at (0.5,-1.75) [label=below: \scriptsize{$7$}] {};
  \node[vertex2] (v8) at (1.5,0.6) [label=above: \scriptsize{$8$}] {};
  \node[vertex2] (v9) at (-1.25,-1.25) [label={[xshift=6pt, yshift=-7pt] \scriptsize{$9$}}] {};
  \node[vertex2] (v10) at (1.5,-1.75) [label=below: \scriptsize{$10$}] {};
  \node[vertex2] (v11) at (1.1,0.1) [label={[xshift=6pt, yshift=-7pt] \scriptsize{$11$}}] {};
  \node[vertex2] (v12) at (-1,-0.5) [label={[xshift=7pt, yshift=-11pt] \scriptsize{$12$}}] {};
  \node[vertex2] (v13) at (1,-0.75) [label={[xshift=6pt, yshift=-4pt] \scriptsize{$13$}}] {};
  \node[vertex2] (v14) at (0.5,1) [label=above: \scriptsize{$14$}] {};
  \node[vertex2] (v15) at (-1.5,0.8) [label=above: \scriptsize{$15$}] {};
  \node[vertex2] (v16) at (-3,-1.5) [label=below: \scriptsize{$16$}] {};
  \node[vertex2] (v17) at (-0.75,0.75) [label=above: \scriptsize{$17$}] {};
  \node[vertex2] (v18) at (-2.5,-0.25) [label=left: \scriptsize{$18$}] {};
  \node[vertex2] (v19) at (-2.5,-1) [label={[xshift=-7pt, yshift=-7pt] \scriptsize{$19$}}] {};
  \node[vertex2] (v20) at (-0.75,-1.75) [label=below: \scriptsize{$20$}] {};
  
  \draw[thick, green] (v1) -- (v6);
  \draw[thick, red] (v1) -- (v17);
  \draw[thick, red] (v1) -- (v19);
  
  \draw[thick, green] (v2) -- (v5);
  \draw[thick, green] (v2) -- (v6);
  \draw[thick, red] (v2) -- (v7);
  
  \draw[very thick, red] (v3) -- (v5);
  \draw[very thick, red] (v3) -- (v8);
  
  \draw[thick, green] (v4) -- (v18);
  
  \draw[thick, red] (v5) -- (v11);
  \draw[thick, red] (v5) -- (v13);
  
  \draw[thick, red] (v6) -- (v11);
  \draw[thick, red] (v6) -- (v12);
  \draw[thick, green] (v6) -- (v13);
  
  \draw[thick, green] (v7) -- (v10);
  \draw[thick, red] (v7) -- (v20);
  
  \draw[thick, green] (v8) -- (v14);
  
  \draw[thick, red] (v9) -- (v20);
  
  \draw[thick, red] (v11) -- (v14);
  
  \draw[thick, red] (v12) -- (v15);
  \draw[thick, red] (v12) -- (v17);
  \draw[thick, green] (v12) -- (v19);
  
  \draw[thick, red] (v14) -- (v17);
  
  \draw[thick, green] (v16) -- (v19);
  
  \draw[thick, red] (v18) -- (v19);
  
  \draw[thick, red] (v19) -- (v20);
  \end{scope}
  
  \begin{scope}[shift={(4.125,-4)}, scale=0.6]
  \node[vertex2] (v1) at (-1.75,0) [label={[xshift=2pt, yshift=-2pt] \scriptsize{$1$}}] {};
  \node[vertex2] (v2) at (0.75,-1.25) [label=left: \scriptsize{$2$}] {};
  \node[vertex2] (v3) at (2.25,0.25) [label=above: \scriptsize{$3$}] {};
  \node[vertex2] (v4) at (-2.75,0.75) [label=above: \scriptsize{$4$}] {};
  \node[vertex2] (v5) at (2,-0.75) [label=right: \scriptsize{$5$}] {};
  \node[vertex2] (v6) at (0,0) [label=above: \scriptsize{$6$}] {};
  \node[vertex2] (v7) at (0.5,-1.75) [label=below: \scriptsize{$7$}] {};
  \node[vertex2] (v8) at (1.5,0.6) [label=above: \scriptsize{$8$}] {};
  \node[vertex2] (v9) at (-1.25,-1.25) [label={[xshift=6pt, yshift=-7pt] \scriptsize{$9$}}] {};
  \node[vertex2] (v10) at (1.5,-1.75) [label=below: \scriptsize{$10$}] {};
  \node[vertex2] (v11) at (1.1,0.1) [label={[xshift=6pt, yshift=-7pt] \scriptsize{$11$}}] {};
  \node[vertex2] (v12) at (-1,-0.5) [label={[xshift=7pt, yshift=-11pt] \scriptsize{$12$}}] {};
  \node[vertex2] (v13) at (1,-0.75) [label={[xshift=6pt, yshift=-4pt] \scriptsize{$13$}}] {};
  \node[vertex2] (v14) at (0.5,1) [label=above: \scriptsize{$14$}] {};
  \node[vertex2] (v15) at (-1.5,0.8) [label=above: \scriptsize{$15$}] {};
  \node[vertex2] (v16) at (-3,-1.5) [label=below: \scriptsize{$16$}] {};
  \node[vertex2] (v17) at (-0.75,0.75) [label=above: \scriptsize{$17$}] {};
  \node[vertex2] (v18) at (-2.5,-0.25) [label=left: \scriptsize{$18$}] {};
  \node[vertex2] (v19) at (-2.5,-1) [label={[xshift=-7pt, yshift=-7pt] \scriptsize{$19$}}] {};
  \node[vertex2] (v20) at (-0.75,-1.75) [label=below: \scriptsize{$20$}] {};
  
  \draw[thick, blue] (v1) -- (v4);
  \draw[thick, green] (v1) -- (v6);
  \draw[thick, blue] (v1) -- (v9);
  \draw[thick, blue] (v1) -- (v18);
  
  \draw[thick, green] (v2) -- (v5);
  \draw[thick, green] (v2) -- (v6);
  \draw[thick, blue] (v2) -- (v13);
  
  \draw[thick, green] (v4) -- (v18);
  
  \draw[thick, blue] (v5) -- (v10);
  
  \draw[thick, blue] (v6) -- (v8);
  \draw[thick, green] (v6) -- (v13);
  \draw[thick, blue] (v6) -- (v15);
  
  \draw[thick, green] (v7) -- (v10);

  \draw[thick, green] (v8) -- (v14);
  
  \draw[thick, blue] (v9) -- (v12);

  \draw[thick, green] (v12) -- (v19);
  
  \draw[thick, blue] (v15) -- (v18);
  
  \draw[thick, green] (v16) -- (v19);
  \end{scope}

  \begin{scope}[shift={(-4.125,-7.5)}, scale=0.6]
  \node[vertex2_blue] (v1) at (-1.75,0) [label={[xshift=2pt, yshift=-2pt] \scriptsize{$1$}}] {};
  \node[vertex2_blue] (v2) at (0.75,-1.25) [label=left: \scriptsize{$2$}] {};
  \node[vertex2_deleted] (v3) at (2.25,0.25) [label=above: \scriptsize{$3$}] {};
  \node[vertex2_red] (v4) at (-2.75,0.75) [label=above: \scriptsize{$4$}] {};
  \node[vertex2_deleted] (v5) at (2,-0.75) [label=right: \scriptsize{$5$}] {};
  \node[vertex2_deleted] (v6) at (0,0) [label=above: \scriptsize{$6$}] {};
  \node[vertex2_red] (v7) at (0.5,-1.75) [label=below: \scriptsize{$7$}] {};
  \node[vertex2_deleted] (v8) at (1.5,0.6) [label=above: \scriptsize{$8$}] {};
  \node[vertex2_deleted] (v9) at (-1.25,-1.25) [label={[xshift=6pt, yshift=-7pt] \scriptsize{$9$}}] {};
  \node[vertex2_deleted] (v10) at (1.5,-1.75) [label=below: \scriptsize{$10$}] {};
  \node[vertex2_red] (v11) at (1.1,0.1) [label={[xshift=6pt, yshift=-7pt] \scriptsize{$11$}}] {};
  \node[vertex2_red] (v12) at (-1,-0.5) [label={[xshift=7pt, yshift=-11pt] \scriptsize{$12$}}] {};
  \node[vertex2_blue] (v13) at (1,-0.75) [label={[xshift=6pt, yshift=-4pt] \scriptsize{$13$}}] {};
  \node[vertex2_red] (v14) at (0.5,1) [label=above: \scriptsize{$14$}] {};
  \node[vertex2_red] (v15) at (-1.5,0.8) [label=above: \scriptsize{$15$}] {};
  \node[vertex2_deleted] (v16) at (-3,-1.5) [label=below: \scriptsize{$16$}] {};
  \node[vertex2_blue] (v17) at (-0.75,0.75) [label=above: \scriptsize{$17$}] {};
  \node[vertex2_red] (v18) at (-2.5,-0.25) [label=left: \scriptsize{$18$}] {};
  \node[vertex2_deleted] (v19) at (-2.5,-1) [label={[xshift=-7pt, yshift=-7pt] \scriptsize{$19$}}] {};
  \node[vertex2_deleted] (v20) at (-0.75,-1.75) [label=below: \scriptsize{$20$}] {};
  
  \draw[thick] (v1) -- (v4);
  \draw[thick] (v1) -- (v17);
  \draw[thick] (v1) -- (v18);
  
  \draw[thick] (v2) -- (v7);
  \draw[thick] (v2) -- (v13);
  
  \draw[thick] (v4) -- (v18);
  
  \draw[thick] (v11) -- (v14);
  
  \draw[thick] (v12) -- (v15);
  \draw[thick] (v12) -- (v17);
  
  \draw[thick] (v14) -- (v17);
  
  \draw[thick] (v15) -- (v18);
  \end{scope}
  
  \begin{scope}[shift={(0,-7.5)}, scale=0.6]
  \node[vertex2_deleted] (v1) at (-1.75,0) [label={[xshift=2pt, yshift=-2pt] \scriptsize{$1$}}] {};
  \node[vertex2_deleted] (v2) at (0.75,-1.25) [label=left: \scriptsize{$2$}] {};
  \node[vertex2_green] (v3) at (2.25,0.25) [label=above: \scriptsize{$3$}] {};
  \node[vertex2_red] (v4) at (-2.75,0.75) [label=above: \scriptsize{$4$}] {};
  \node[vertex2_green] (v5) at (2,-0.75) [label=right: \scriptsize{$5$}] {};
  \node[vertex2_green] (v6) at (0,0) [label=above: \scriptsize{$6$}] {};
  \node[vertex2_red] (v7) at (0.5,-1.75) [label=below: \scriptsize{$7$}] {};
  \node[vertex2_green] (v8) at (1.5,0.6) [label=above: \scriptsize{$8$}] {};
  \node[vertex2_green] (v9) at (-1.25,-1.25) [label={[xshift=6pt, yshift=-7pt] \scriptsize{$9$}}] {};
  \node[vertex2_green] (v10) at (1.5,-1.75) [label=below: \scriptsize{$10$}] {};
  \node[vertex2_red] (v11) at (1.1,0.1) [label={[xshift=6pt, yshift=-7pt] \scriptsize{$11$}}] {};
  \node[vertex2_red] (v12) at (-1,-0.5) [label={[xshift=7pt, yshift=-11pt] \scriptsize{$12$}}] {};
  \node[vertex2_deleted] (v13) at (1,-0.75) [label={[xshift=6pt, yshift=-4pt] \scriptsize{$13$}}] {};
  \node[vertex2_red] (v14) at (0.5,1) [label=above: \scriptsize{$14$}] {};
  \node[vertex2_red] (v15) at (-1.5,0.8) [label=above: \scriptsize{$15$}] {};
  \node[vertex2_green] (v16) at (-3,-1.5) [label=below: \scriptsize{$16$}] {};
  \node[vertex2_deleted] (v17) at (-0.75,0.75) [label=above: \scriptsize{$17$}] {};
  \node[vertex2_red] (v18) at (-2.5,-0.25) [label=left: \scriptsize{$18$}] {};
  \node[vertex2_green] (v19) at (-2.5,-1) [label={[xshift=-7pt, yshift=-7pt] \scriptsize{$19$}}] {};
  \node[vertex2_green] (v20) at (-0.75,-1.75) [label=below: \scriptsize{$20$}] {};
  
  \draw[thick] (v3) -- (v5);
  \draw[thick] (v3) -- (v8);
  
  \draw[thick] (v4) -- (v18);
  
  \draw[thick] (v5) -- (v10);
  \draw[thick] (v5) -- (v11);
  
  \draw[thick] (v6) -- (v8);
  \draw[thick] (v6) -- (v11);
  \draw[thick] (v6) -- (v12);
  \draw[thick] (v6) -- (v15);
  
  \draw[thick] (v7) -- (v10);
  \draw[thick] (v7) -- (v20);
  
  \draw[thick] (v8) -- (v14);
  
  \draw[thick] (v9) -- (v12);
  \draw[thick] (v9) -- (v20);
  
  \draw[thick] (v11) -- (v14);
  
  \draw[thick] (v12) -- (v15);
  \draw[thick] (v12) -- (v19);
  
  \draw[thick] (v15) -- (v18);
  
  \draw[thick] (v16) -- (v19);
  
  \draw[thick] (v18) -- (v19);
  
  \draw[thick] (v19) -- (v20);
  \end{scope}
  
  \begin{scope}[shift={(4.125,-7.5)}, scale=0.6]
  \node[vertex2_blue] (v1) at (-1.75,0) [label={[xshift=2pt, yshift=-2pt] \scriptsize{$1$}}] {};
  \node[vertex2_blue] (v2) at (0.75,-1.25) [label=left: \scriptsize{$2$}] {};
  \node[vertex2_green] (v3) at (2.25,0.25) [label=above: \scriptsize{$3$}] {};
  \node[vertex2_deleted] (v4) at (-2.75,0.75) [label=above: \scriptsize{$4$}] {};
  \node[vertex2_green] (v5) at (2,-0.75) [label=right: \scriptsize{$5$}] {};
  \node[vertex2_green] (v6) at (0,0) [label=above: \scriptsize{$6$}] {};
  \node[vertex2_deleted] (v7) at (0.5,-1.75) [label=below: \scriptsize{$7$}] {};
  \node[vertex2_green] (v8) at (1.5,0.6) [label=above: \scriptsize{$8$}] {};
  \node[vertex2_green] (v9) at (-1.25,-1.25) [label={[xshift=6pt, yshift=-7pt] \scriptsize{$9$}}] {};
  \node[vertex2_green] (v10) at (1.5,-1.75) [label=below: \scriptsize{$10$}] {};
  \node[vertex2_deleted] (v11) at (1.1,0.1) [label={[xshift=6pt, yshift=-7pt] \scriptsize{$11$}}] {};
  \node[vertex2_deleted] (v12) at (-1,-0.5) [label={[xshift=7pt, yshift=-11pt] \scriptsize{$12$}}] {};
  \node[vertex2_blue] (v13) at (1,-0.75) [label={[xshift=6pt, yshift=-4pt] \scriptsize{$13$}}] {};
  \node[vertex2_deleted] (v14) at (0.5,1) [label=above: \scriptsize{$14$}] {};
  \node[vertex2_deleted] (v15) at (-1.5,0.8) [label=above: \scriptsize{$15$}] {};
  \node[vertex2_green] (v16) at (-3,-1.5) [label=below: \scriptsize{$16$}] {};
  \node[vertex2_blue] (v17) at (-0.75,0.75) [label=above: \scriptsize{$17$}] {};
  \node[vertex2_deleted] (v18) at (-2.5,-0.25) [label=left: \scriptsize{$18$}] {};
  \node[vertex2_green] (v19) at (-2.5,-1) [label={[xshift=-7pt, yshift=-7pt] \scriptsize{$19$}}] {};
  \node[vertex2_green] (v20) at (-0.75,-1.75) [label=below: \scriptsize{$20$}] {};
  
  \draw[thick] (v1) -- (v6);
  \draw[thick] (v1) -- (v9);
  \draw[thick] (v1) -- (v17);
  \draw[thick] (v1) -- (v19);
  
  \draw[thick] (v2) -- (v5);
  \draw[thick] (v2) -- (v6);
  \draw[thick] (v2) -- (v13);
  
  \draw[thick] (v3) -- (v5);
  \draw[thick] (v3) -- (v8);
  
  \draw[thick] (v5) -- (v10);
  \draw[thick] (v5) -- (v13);
  
  \draw[thick] (v6) -- (v8);
  \draw[thick] (v6) -- (v13);
  
  \draw[thick] (v9) -- (v20);
  
  \draw[thick] (v16) -- (v19);
  
  \draw[thick] (v19) -- (v20);
  \end{scope}
  \end{tikzpicture}
  \caption[caption]{Illustrating the color-avoiding edge- and vertex-connectivity on an Erd\H{o}s-R\'enyi random graph $G(20, 0.15)$. The first row depicts an edge-colored and a vertex-colored graph with three equiprobable colors. The second and third rows show the graphs after removing the green, blue and red edges/vertices, respectively.
    \\\hspace{\textwidth}
    Note that vertex 5 and vertex 14 are color-avoiding edge-connected since the path 5-11-14 serves as both a green-avoiding and a blue-avoiding path, while the path 5-2-6-8-14 serves as a red-avoiding path. On the other hand, vertex 16 is not color-avoiding edge-connected to any other vertices since without green edges it gets isolated.
    \\\hspace{\textwidth}
    Considering color-avoiding vertex-connectivity, vertex 6 and vertex 10 are strongly (and therefore weakly) color-avoiding connected since the path 6-11-5-10 serves as a blue-avoiding path, the path 6-13-5-10 serves as a red-avoiding path, while the path 6-2-7-10 serves as a green-avoiding path for internal vertices -- this last condition is only required for strong color-avoiding connectivity. Vertex 3 and vertex 10 are weakly color-avoiding connected since the path 3-5-10 serves as both a red-avoiding and blue-avoiding path, while they are not strongly color-avoiding connected since no path from vertex 3 can avoid green vertices as internal nodes. Vertex 7 and vertex 11 are not strongly/weakly color-avoiding connected since no green-avoiding path exists between them.}
    \label{fig:colored_graph}
\end{figure}
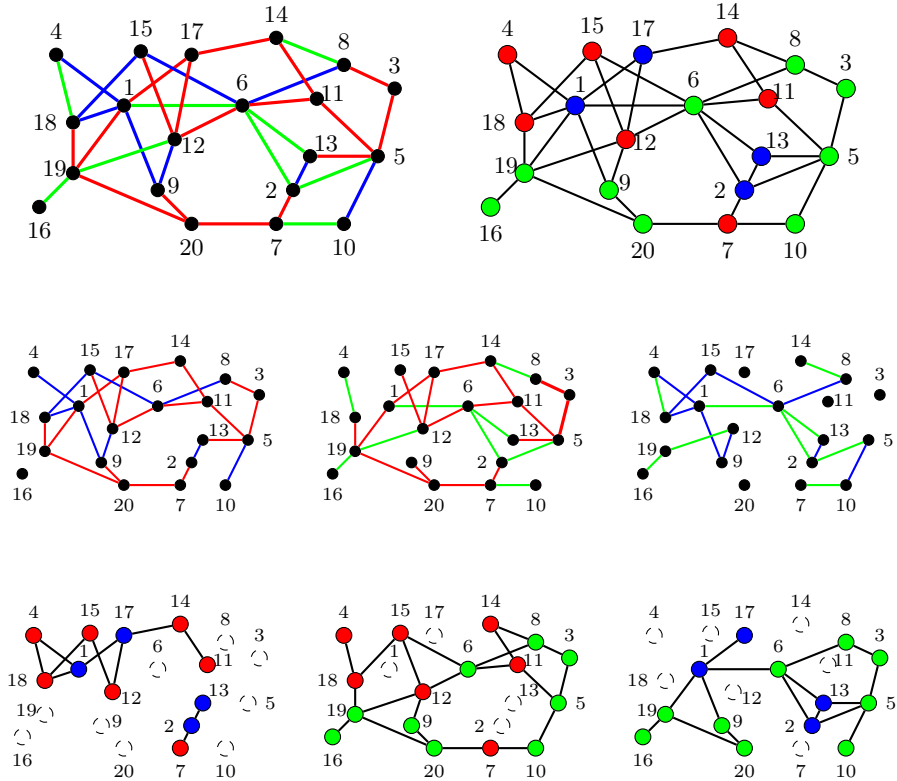

\subsubsection{Every edge has one color}

\begin{definition}
\label{def:edge_con}
Let $G$ be a graph, $C = \{ c_1, \ldots, c_k \}$ a color set and $c: \ E(G) \to C$ a function that assign colors to the edges.

For any $i \in \{ 1, \ldots, k \}$ the vertices $u,v \in V(G)$ are called \textbf{$\boldsymbol{c_i}$-avoiding edge-connected} if after the removal of the edges of color $c_i$, $u$ and $v$ are in the same component in the remaining graph, i.e there exists a path between $u$ and $v$ which does not contain any edges of color $c_i$ -- such a path is called a \textbf{$\boldsymbol{c_i}$-avoiding path}.

We say that the vertices $u,v \in V(G)$ are \textbf{color-avoiding edge-connected} if they are $c_i$-avoiding edge-connected for all $i \in \{ 1, \ldots, k\}$.
\end{definition}




The relation of color-avoiding edge-connectivity (see Fig. \ref{edgecoloring_example}) is an equivalence relation and thus it defines a partition of the vertex set. The equivalence classes are called \emph{color-avoiding edge-connected components}.

 \begin{figure}
  \centering
  \begin{tikzpicture}
  \tikzstyle{vertex}=[draw,circle,fill,minimum size=5,inner sep=0]
  
  \node[vertex] (v1) at (90:0.5) [label=above: $v_1$] {};
  \node[vertex] (v2) at (210:0.5) [label=below left: $v_2$] {};
  \node[vertex] (v3) at (330:0.5) [label=below right: $v_3$] {};
  
  \draw[very thick, blue] (v1) -- (v2) -- (v3);
  \draw[very thick, red] (v1) -- (v3);
  \end{tikzpicture}
  \caption{The vertices $v_2$ and $v_3$ are color-avoiding edge-connected. But $v_1$ and $v_2$ are not: the removal of the blue edges disconnects them.}
  \label{edgecoloring_example}
 \end{figure}
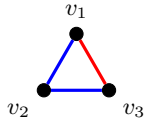

\subsubsection{Multiple edges}

A natural modification of the previous definition is if we allow for multiple edges, meaning that there can be more types of connection between two nodes (see Fig. \ref{multiple_edges}). Multiple edges make the network less vulnerable: if there are at least two edges of different colors 
 between two nodes, then they are color-avoiding edge-connected since their connection cannot be destroyed by attacking only one color at a time.
 
 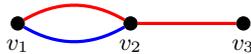
\begin{figure}
  \centering
  \begin{tikzpicture}
  \tikzstyle{vertex}=[draw,circle,fill,minimum size=5,inner sep=0]
  
  \node[vertex] (v1) at (0,0) [label=below: $v_1$] {};
  \node[vertex] (v2) at (1.5,0) [label=below: $v_2$] {};
  \node[vertex] (v3) at (3,0) [label=below: $v_3$] {};
  
  \draw[very thick, red] (v1) to [bend left=30] (v2);
  \draw[very thick, blue] (v1) to [bend right=30] (v2);
  \draw[very thick, red] (v2) -- (v3);
  \end{tikzpicture}
  \caption{The vertices $v_1$ and $v_2$ are color-avoiding edge-connected. But $v_2$ and $v_3$ are not: the removal of the red edges disconnects them.}
  \label{multiple_edges}
 \end{figure}

\subsubsection{Every edge has a list of colors}

Another possible generalization of the framework is to modify Def. \ref{def:edge_con} in such a way that we make the edges more sensitive to attack by assigning a list of colors to the vertices representing all the vulnerabilities that an edge has and an edge is destroyed whenever one of its colors is attacked (see Fig. \ref{edges_lists}). Formally, the function $c$ is modified:
$c: \ E(G) \to 2^C$, where $2^C$ is the power set of the color set $C=\{ c_1, \ldots, c_k \}$. Furthermore, in this scenario we say that the vertices $u,v \in V(G)$ are $c_i$-avoiding edge-connected if after the removal of the edges that contain $c_i$ on their lists of colors, $u$ and $v$ are in the same component in the remaining graph.

%
%
 
 \begin{figure}
  \centering
  \begin{tikzpicture}
  \tikzstyle{vertex}=[draw,circle,fill,minimum size=5,inner sep=0]
  
  \node[vertex] (v1) at (90:1.5) [label=above: $v_1$] {};
  \node[vertex] (v2) at (180:1.5) [label=left: $v_2$] {};
  \node[vertex] (v3) at (270:1.5) [label=below: $v_3$] {};
  \node[vertex] (v4) at (0:1.5) [label=right: $v_4$] {};
  
  \draw[very thick] (v1) -- (v2) node[pos=0.5, above, sloped] {\small{\textcolor{red}{red}, \textcolor{blue}{blue}}};
  \draw[very thick] (v1) -- (v4) node[pos=0.5, above, sloped] {\small{\textcolor{blue}{blue}}};
  \draw[very thick] (v2) -- (v3) node[pos=0.5, below, sloped] {\small{\textcolor{red}{red}, \textcolor{green}{green}}};
  \draw[very thick] (v2) -- (v4) node[pos=0.5, above, sloped] {\small{\textcolor{blue}{blue}, \textcolor{green}{green}}};
  \draw[very thick] (v3) -- (v4) node[pos=0.5, below, sloped] {\small{\textcolor{green}{green}}};
  \end{tikzpicture}
  \caption{The vertices $v_2$ and $v_4$ are color-avoiding edge-connected: a red-, blue- and green-avoiding path between them are $v_2 v_4$, $v_2 v_3 v_4$ and $v_2 v_1 v_4$, respectively. But $v_1$ and $v_2$ are not: the removal of the edges containing the color blue on their lists disconnects them.}
  \label{edges_lists}
 \end{figure}
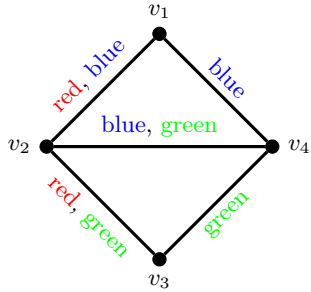

\subsection{Coloring the vertices}

It is also interesting to consider the case when the colors are assigned to vertices that are exposed to attack or failure while the edges  remain indistinguishable and unharmed. Possible real-world scenarios of having vulnerable classes of nodes include AS-level Internet with routers registered in different countries, telecommunication networks with transmission towers operated by different providers. The color of the nodes can also represent e.g. ownership, geographical location, dependence on a critical material \cite{krause2016hidden}. The strong/weak vertex color-avoiding connectivity is illustrated in Fig. \ref{fig:colored_graph} on an Erd\H{o}s-R\'enyi random graph.

\subsubsection{Strong color-avoiding connectivity}

\begin{definition}[\cite{krause2016hidden}]
\label{def:strong}
Let $G$ be a graph, $C = \{ c_1, \ldots, c_k \}$ a color set and $c: V(G) \to C$ a function that assigns colors to the vertices.

For any $i \in \{ 1, \ldots, k \}$ the vertices $u, v \in V(G)$ are called \textbf{strongly $\boldsymbol{c_i}$-avoiding vertex-connected} (or \textbf{strongly $\boldsymbol{c_i}$-avoiding connected}) if after the removal of the vertices of color $c_i$ excluding these two vertices, they are in the same component in the remaining graph, i.e.~there exists a path between $u$ and $v$ whose internal vertices are not of color $c_i$ -- such a path is called a \textbf{$\boldsymbol{c_i}$-avoiding path}.

We say that the vertices $u,v \in V(G)$ are \textbf{strongly color-avoiding vertex-connected} (or \textbf{strongly color-avoiding connected}) if they are strongly $c_i$-avoiding connected for all $i \in \{ 1, \ldots, k \}$.
\end{definition}

In this case the relation of strong color-avoiding connectivity is not transitive, therefore it is not an equivalence relation (see Fig. \ref{weak&strong_example}). The \emph{strongly color-avoiding connected components} are maximal sets of vertices such that any two of them are strongly color-avoiding connected.

\subsubsection{More colors on the vertices} Similarly to the case when the edges were colored, here we can also assign multiple colors to vertices. One can think of the multiple colors (lists of colors) as multiple vulnerabilities on the nodes making the network less robust. Another approach is to consider the scenario analogously to multiple edges. Here it is important to note that if a node has at least two different colors, it makes the vertex immortal under the previously mentioned color attacks analogously to the "multiple edges" approach.

\subsubsection{Weak color-avoiding connectivity}

Contrary to color-avoiding edge-con\-nec\-tiv\-ity, if we color the vertices, it is dubious how to handle the source and target nodes in the definition of color-avoiding vertex-connectivity. In Def. \ref{def:strong} a possible approach was presented that can capture several realistic scenarios. Considering eavesdropping, it is reasonable that the sender and receiver guarantee the security of the message but vulnerability may affect the nodes as transmitters. However, when the attack of the vertices rather means destroying the entities, it is more natural to consider another approach to define color-avoiding vertex-connectivity: attacking red vertices makes it pointless which nodes a red vertex can reach on a path without other red vertices. This scenario is captured by the concept of weak color-avoiding connectivity (see Fig. \ref{weak&strong_example}).

\begin{definition}
\label{def:weak}
Let $G$ be a connected graph, $C=\{ c_1, \ldots, c_k \}$ a color set and $c:~ V(G) \to C$ a function that assigns colors to the vertices. 

For any $i \in \{ 1, \ldots, k \}$ the vertices $u, v \in V(G)$ are called \textbf{weakly $\boldsymbol{c_i}$-avoiding vertex-connected} (or \textbf{weakly $\boldsymbol{c_i}$-avoiding connected}) if after the removal of the vertices of color $c_i$, either at least one of $u$ or $v$ is deleted or they are in the same component in the remaining graph, i.e.~if neither $u$ nor $v$ are of color $c_i$, there exists a path between them whose vertices are not of color $c_i$ -- such a path is called a \textbf{$\boldsymbol{c_i}$-avoiding path}.

We say that the vertices $u, v \in V(G)$ are \textbf{weakly color-avoiding vertex-connected} (or \textbf{weakly color-avoiding connected}) if they are weakly $c_i$-avoid\-ing connected for all $i \in \{ 1, \ldots, k \}$.
\end{definition}

Similarly to the strong case, weak color-avoiding connectivity is also not an equivalence relation (see Fig. \ref{weak&strong_example}). The \emph{weakly color-avoiding connected components} are maximal sets of vertices such that any two of them are weakly color-avoiding connected.

We can extend the definition to non-connected graphs with the extra condition that two vertices can be weakly color-avoiding connected only if they are in the same component in the original graph.

\begin{remark}
The notion of weakly color-avoiding connectivity is indeed a weaker concept than the one defined in Def. \ref{def:strong}. It is easy to see that if two vertices are strongly color-avoiding connected then it implies that they are weakly color-avoiding connected as well.
\end{remark}

\begin{figure}
  \centering
  \begin{tikzpicture}
  \tikzstyle{vertex_red}=[draw,circle,fill=red,minimum size=10,inner sep=0]
  \tikzstyle{vertex_blue}=[draw,circle,fill=blue,minimum size=10,inner sep=0]
  
  \node[vertex_red] (v1) at (0,0) [label=below: $v_1$] {};
  \node[vertex_blue] (v2) at (1,0) [label=below: $v_2$] {};
  \node[vertex_red] (v3) at (2,0) [label=below: $v_3$] {};
  \node[vertex_red] (v4) at (3,0) [label=below: $v_4$] {};
  \node[vertex_red] (v5) at (4,0) [label=below: $v_5$] {};
  
  \draw[thick] (v1) -- (v2) -- (v3) -- (v4) -- (v5);
  \end{tikzpicture}
  \caption[caption]{The vertices $v_1$ and $v_2$ are weakly/strongly color-avoiding connected, and so are the vertices $v_2$ and $v_3$. But $v_1$ and $v_3$ are not: the removal of the blue vertices (i.e. the removal of $v_2$) disconnects them. The vertices $v_3$ and $v_5$ are weakly but not strongly color-avoiding connected.\\\hspace{\textwidth}
  The above observation also shows that neither the strong nor the weak color-avoiding connectivity is a transitive relation.}
  \label{weak&strong_example}
 \end{figure}
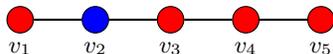

\subsubsection{More colors on the vertices}

Weak color-avoiding connectivity can be extended to multiple colors as well, in the exact same manner as strong color-avoiding connectivity.

\subsubsection{Other generalizations}

It is worth mentioning that other generalizations have been also proposed. Krause et al. \cite{krause2017color} consider nodes with differentiated functions, either as senders/receivers or transmitters. They introduce a flexible trust scenario where vertices can be trusted or avoided in both functions. Trusting colors for transmission naturally increases color-avoiding connectivity \cite{krause2017color}.

\section{Computational complexity of finding the color-avoiding components}
\label{complex}

After presenting the problem of color-avoiding percolation and various modeling approaches, in this section we analyze the computational complexity of finding the robustly (color-avoiding) connected components considering the different problem definitions. In the following we
assume the reader's acquaintance with standard concepts of computational
complexity theory that may be found e.g. in \cite{sipser2012introduction}. We will use in this section that the following well-known decision problem is NP-complete \cite{sipser2012introduction}.

\medskip

\noindent {\bf \scshape Clique} \\
 \textit{Instance:} a graph $G$ and a positive integer $l$. \\
 \textit{Question:} does $G$ have a clique of size at least $l$?
 
 \medskip

Now we list the decision problems for which the computational complexity will be presented in this section. Although the problems are seemingly very similar, they differ considerably concerning their complexity. 

\medskip

\noindent {\bf \scshape ColorAvoidingEdgeConnectedComponent} \\
 \textit{Instance:} a graph $G$, a color set $C=\{ c_1, \ldots, c_k \}$, a function $c:~E(G) \to C$ and a positive integer $l$. \\
 \textit{Question:} is it true that $G$ has a color-avoiding edge-connected component of size at least $l$?
 
 \medskip
 
 \noindent {\bf \scshape StronglyColorAvoidingConnectedComponent} \\
 \textit{Instance:} a graph $G$, a color set $C=\{ c_1, \ldots, c_k \}$, a function $c:~V(G) \to C$ and a positive integer $l$. \\
 \textit{Question:} is it true that $G$ has a strongly color-avoiding connected component of size at least $l$?
 
 \medskip
 
 \noindent {\bf \scshape WeaklyColorAvoidingConnectedComponent} \\
 \textit{Instance:} a graph $G$, a color set $C=\{ c_1, \ldots, c_k \}$, a function $c:~V(G) \to C$ and a positive integer $l$. \\
 \textit{Question:} is it true that $G$ has a weakly color-avoiding connected component of size at least $l$?

\medskip
 
 \noindent {\bf \scshape WeaklyColorAvoidingConnectedComponent-ListOfColors} \\
 \textit{Instance:} a graph $G$, a color set $C=\{ c_1, \ldots, c_k \}$, a function $c:~V(G) \to 2^C$ and a positive integer $l$. \\
 \textit{Question:} is it true that $G$ has a weakly color-avoiding connected component of size at least $l$?

\medskip

First we prove that the color-avoiding edge-connected components can be found in polynomial time.

\begin{theorem}
 The problem {\scshape ColorAvoidingEdgeConnectedComponent} is in P. More precisely, the color-avoiding edge-connected components of $G$ can be found in polynomial time.
\end{theorem}
\begin{proof}
 Let $G'$ be a graph on the vertex set of $G$ where two vertices are connected if and only if they are color-avoiding edge-connected (for an example see Fig. \ref{G'_edgecoloring}). Obviously, $G'$ can be constructed in polynomial time: we need to check for every pair of vertices whether they remain in the same component after erasing the edges of each color separately.
 
 \begin{figure}
  \centering
  \begin{tikzpicture}
  \tikzstyle{vertex}=[draw,circle,fill,minimum size=5,inner sep=0]
  
  \node at (-1.5,1.25) {$G$};
  
  \node[vertex] (v1) at (90:0.75) [label=above: $v_1$] {};
  \node[vertex] (v2) at (150:0.75) [label=left: $v_2$] {};
  \node[vertex] (v3) at (210:0.75) [label=left: $v_3$] {};
  \node[vertex] (v4) at (270:0.75) [label=below: $v_4$] {};
  \node[vertex] (v5) at (330:0.75) [label=right: $v_5$] {};
  \node[vertex] (v6) at (30:0.75) [label=right: $v_6$] {};
  
  \draw[very thick, red] (v1) -- (v2);
  \draw[very thick, blue] (v1) -- (v3) -- (v2);
  \draw[very thick, blue] (v4) -- (v5);
  \draw[very thick, red] (v4) -- (v6) -- (v5);
  \draw[very thick, red] (v1) -- (v6);
  \draw[very thick, red] (v3) -- (v4);
  
  \begin{scope}[shift={(5,0)}]
  \node at (-1.5,1.25) {$G'$};
  
  \node[vertex] (v1) at (90:0.75) [label=above: $v_1$] {};
  \node[vertex] (v2) at (150:0.75) [label=left: $v_2$] {};
  \node[vertex] (v3) at (210:0.75) [label=left: $v_3$] {};
  \node[vertex] (v4) at (270:0.75) [label=below: $v_4$] {};
  \node[vertex] (v5) at (330:0.75) [label=right: $v_5$] {};
  \node[vertex] (v6) at (30:0.75) [label=right: $v_6$] {};
  
  \draw[very thick] (v1) -- (v2) -- (v3) -- (v1);
  \draw[very thick] (v4) -- (v5);
  \end{scope}
  \end{tikzpicture}
  \caption{Two vertices are adjacent in $G'$ if and only if they are color-avoiding edge-connected in $G$.}
  \label{G'_edgecoloring}
 \end{figure}
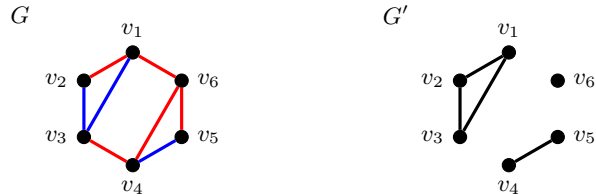
 
 Since the color-avoiding edge-connectivity is an equivalence relation, the graph $G'$ is $P_3$-free, i.e.~it cannot contain a path on 3 vertices as an induced subgraph. Obviously, the color-avoiding edge-connected components of $G$ are exactly the maximal cliques of $G'$. 
 
 It is easy to see that the components of a $P_3$-free graph are cliques, therefore the maximal cliques of $G'$ are its components. Hence, the color-avoiding edge-connected components of $G$ can be found in polynomial time.
\end{proof}

The above theorem obviously can be applied when there are multiple edges or when lists of colors are associated with the edges. Clearly, the same proof works in both cases.

Now, we move on to the analysis of color-avoiding vertex percolation. First, we prove that the stronger definition (Def. \ref{def:strong}) leads to an NP-complete problem.

\begin{theorem}
 \ The \ problem \ {\scshape StronglyColorAvoidingConnectedComponent} is NP-complete.
\end{theorem}
\begin{proof}
 Obviously, this problem is in NP: a witness is a strongly color-avoiding connected component of size at least $l$. To show that this problem is NP-hard we reduce {\scshape Clique} to it.
 
 If we use only one color, then by definition the strongly color-avoiding connected components of $G$ are exactly its maximal cliques, therefore our problem is indeed NP-complete.
\end{proof}

Next, we present that using the weak definition of color-avoiding connectivity (Def. \ref{def:weak}) the connected components can be found in polynomial time. The proof consists of two main parts. First, we show that finding the weakly color-avoiding connected components in any graph is equivalent to finding the cliques of an associated locally chordal graph. This together with the fact that cliques can be found in polynomial time in a locally chordal graph gives us the desired result.

\begin{theorem}
 Let $G$ be a graph, $C=\{ c_1, \ldots, c_k \}$ a set of colors and $c:~ V(G) \to C$ a function that assigns colors to the vertices. Let $G'$ be a graph on the vertex set of $G$ where two vertices are connected if and only if they are weakly color-avoiding connected. Then the graph $G'$ is locally chordal, i.e.~the neighborhood of any vertex cannot contain an induced cycle of length at least 4.
 \label{G'locchordal}
\end{theorem}

\noindent For an example on the construction of graph  $G'$ from Theorem \ref{G'locchordal}  see Fig. \ref{G'_weak}.

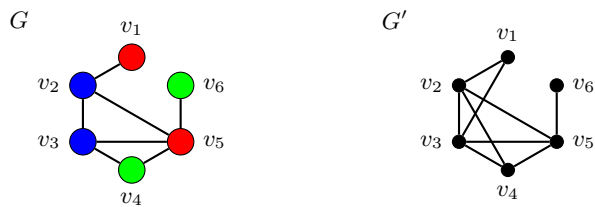
\begin{figure}
  \centering
  \begin{tikzpicture}
  \tikzstyle{vertex}=[draw,circle,fill,minimum size=5,inner sep=0]
  \tikzstyle{vertex_red}=[draw,circle,fill=red,minimum size=10,inner sep=0]
  \tikzstyle{vertex_blue}=[draw,circle,fill=blue,minimum size=10,inner sep=0]
   \tikzstyle{vertex_green}=[draw,circle,fill=green,minimum size=10,inner sep=0]
  \node at (-1.5,1.25) {$G$};
  
  \node[vertex_red] (v1) at (90:0.75) [label=above: $v_1$] {};
  \node[vertex_blue] (v2) at (150:0.75) [label=left: $v_2$] {};
  \node[vertex_blue] (v3) at (210:0.75) [label=left: $v_3$] {};
  \node[vertex_green] (v4) at (270:0.75) [label=below: $v_4$] {};
  \node[vertex_red] (v5) at (330:0.75) [label=right: $v_5$] {};
  \node[vertex_green] (v6) at (30:0.75) [label=right: $v_6$] {};
  
  \draw[thick] (v1) -- (v2) -- (v3) -- (v4) -- (v5) -- (v6);
  \draw[thick] (v3) -- (v5) -- (v2);
  
  \begin{scope}[shift={(5,0)}]
  \node at (-1.5,1.25) {$G'$};
  
  \node[vertex] (v1) at (90:0.75) [label=above: $v_1$] {};
  \node[vertex] (v2) at (150:0.75) [label=left: $v_2$] {};
  \node[vertex] (v3) at (210:0.75) [label=left: $v_3$] {};
  \node[vertex] (v4) at (270:0.75) [label=below: $v_4$] {};
  \node[vertex] (v5) at (330:0.75) [label=right: $v_5$] {};
  \node[vertex] (v6) at (30:0.75) [label=right: $v_6$] {};
  
  \draw[thick] (v1) -- (v2) -- (v3) -- (v4) -- (v5) -- (v6);
  \draw[thick] (v3) -- (v5) -- (v2);
  \draw[thick] (v1) -- (v3);
  \draw[thick] (v2) -- (v4);
  \end{scope}
  \end{tikzpicture}
  \caption{Two vertices are adjacent in $G'$ if and only if they are weakly color-avoiding connected in $G$.}
  \label{G'_weak}
 \end{figure}

\begin{proof}
We note that throughout this proof the notion "color-avoiding" always stands for "weakly color-avoiding".

 It is easy to see that a graph is locally chordal if and only if it does not contain a wheel on at least five vertices as an induced subgraph: if the graph contains an induced wheel on at least five vertices, then the outer cycle of this wheel is an induced cycle of length at least four in the neighborhood of the center vertex, therefore the graph is not locally chordal. To prove the reverse direction, suppose that the graph is not locally chordal, i.e., there exists a vertex whose neighborhood contains an induced cycle of length at least four. Then this vertex and this cycle together form an induced wheel on at least five vertices.

 Suppose to the contrary that $G'$ contains a wheel on $l+1 \ge 5$ vertices as an induced subgraph. Let $u$ be the center vertex of this wheel, and $w_1, \ldots, w_l$ be the vertices of the outer cycle (in this order), see Fig.~\ref{wheel-free}.
 
 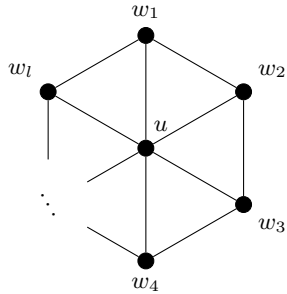
\begin{figure}
  \centering
  \begin{tikzpicture}
  \tikzstyle{vertex}=[draw,circle,fill=black,minimum size=6,inner sep=0]
  
  \node[vertex] (u) at (0,0) [label={[shift={(0.2,0)}] $u$}] {};
  \node[vertex] (w1) at (90:1.5) [label=above: $w_1$] {};
  \node[vertex] (w2) at (30:1.5) [label=above right: $w_2$] {};
  \node[vertex] (w3) at (330:1.5) [label=below right: $w_3$] {};
  \node[vertex] (w4) at (270:1.5) [label=below: $w_4$] {};
  \node (w5) at (210:1.5) {\rotatebox{-60}{$\ldots$}};
  \node[vertex] (w6) at (150:1.5) [label=above left: $w_l$] {};
  
  \draw (u) -- (w1);
  \draw (u) -- (w2);
  \draw (u) -- (w3);
  \draw (u) -- (w4);
  \draw (u) -- ($(u)!0.6!(w5)$);
  \draw (u) -- (w6);
  \draw (w6) -- (w1) -- (w2) -- (w3) -- (w4);
  \draw (w4) -- ($(w4)!0.6!(w5)$);
  \draw (w6) -- ($(w6)!0.6!(w5)$);
  \end{tikzpicture}
  \caption{The wheel on $l+1$ vertices.}
  \label{wheel-free}
 \end{figure}
 
 We can assume that the color of the vertex $w_2$ is $c_1$. Now consider the vertices $w_1$ and $w_3$. Since they are not connected in $G'$, there exists at least one color such that the removal of the vertices of that color disconnects them. On the other hand, the $c_i$-avoiding paths from $w_1$ to $w_2$ and from $w_2$ to $w_3$ (which exist since $w_1 w_2, w_2 w_3 \in E(G')$) can be combined into $c_i$-avoiding paths from $w_1$ to $w_3$ for every color $c_i \in C \setminus \{ c_1 \}$. (Obviously, this procedure does not work with color $c_1$ since the vertex $w_2$ is of color $c_1$.) Thus, only the removal of the vertices of color $c_1$ can disconnect $w_1$ and $w_3$. Therefore, $u$ must have also color $c_1$ (otherwise the $c_1$-avoiding paths from $w_1$ to $w_2$ and from $w_2$ to $w_3$ could be combined into a $c_1$-avoiding path from $w_1$ to $w_3$).
 
 Now consider the vertices $w_2$ and $w_4$. Since they are not connected in $G'$, there exists at least one color such that the removal of the vertices of that color disconnects them. However, the $c_i$-avoiding paths from $w_2$ to $u$ and from $u$ to $w_4$ (which exist since $u w_2, u w_4 \in E(G')$) can be combined into $c_i$-avoiding paths from $w_2$ to $w_4$ for every color $c_i \in C \setminus \{ c_1 \}$. Again, this procedure does not work with color $c_1$ since the vertex $u$ is of color $c_1$. But since $w_2$ is also of color $c_1$, $w_2$ and $w_4$ are weakly $c_1$-avoiding connected by definition. Hence, they are weakly color-avoiding connected, which is a contradiction.
\end{proof}

\begin{theorem}[\cite{gavril1996intersection}]
 The maximal cliques of any locally chordal graph can be found in polynomial time.
\end{theorem}

\begin{corollary}
\ The \ problem \ {\scshape WeaklyColorAvoidingConnectedComponent} is in P. More precisely, the weakly color-avoiding connected components of $G$ can be found in polynomial time.
\end{corollary}

The above theorem obviously can be applied in the more robust case when there may be multiple colors on the vertices resulting in indestructible nodes.

On the other hand, in the other case  -- when the vertices have multiple colors (lists of colors) and a vertex is destroyed whenever one of its colors is attacked -- seemingly paradoxically -- leads to a much harder, NP-complete problem.

\begin{theorem}
 \ The \ {\scshape WeaklyColorAvoidingConnectedComponent-\-List\-Of\-Colors} problem is NP-complete.
 \label{weak_lists_complexity}
\end{theorem}
\begin{proof}
 Obviously, this problem is in NP. To show that this problem is NP-hard we reduce {\scshape Clique} to it.
 
 Assign a color to any two vertices, and add this color to the list of every other vertex (so altogether we use $\binom{n}{2}$ colors and every vertex has $\binom{n}{2} - (n-1)$ colors on its list). For an example on the construction of lists of colors see Fig.~\ref{weak_lists}. Now, two vertices are weakly color-avoiding connected if and only if they are adjacent in $G$. Hence, the weakly color-avoiding connected components of $G$ are exactly its maximal cliques, therefore our problem is indeed NP-complete.
\end{proof}

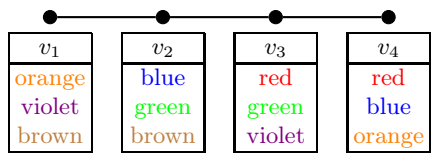
\begin{figure}
  \centering
  \begin{tikzpicture}
  \tikzstyle{vertex}=[draw,circle,fill,minimum size=5,inner sep=0]
  
  \node[vertex] (v1) at (0,0) [label=below: \begin{tabular}{|C|} \hline $v_1$ \\ \hline \small{\textcolor{orange}{orange}} \\ \small{\textcolor{violet}{violet}} \\ \small{\textcolor{brown}{brown}} \\ \hline \end{tabular}] {};
  \node[vertex] (v2) at (1.5,0) [label=below: \begin{tabular}{|C|} \hline $v_2$ \\ \hline \small{\textcolor{blue}{blue}} \\ \small{\textcolor{green}{green}} \\ \small{\textcolor{brown}{brown}} \\ \hline \end{tabular}] {};
  \node[vertex] (v3) at (3,0) [label=below: \begin{tabular}{|C|} \hline $v_3$ \\ \hline \small{\textcolor{red}{red}} \\ \small{\textcolor{green}{green}} \\ \small{\textcolor{violet}{violet}} \\ \hline \end{tabular}] {};
  \node[vertex] (v4) at (4.5,0) [label=below: \begin{tabular}{|C|} \hline $v_4$ \\ \hline \small{\textcolor{red}{red}} \\ \small{\textcolor{blue}{blue}} \\ \small{\textcolor{orange}{orange}} \\ \hline \end{tabular}] {};
  
  \draw[thick] (v1) -- (v2) -- (v3) -- (v4);
  
  \end{tikzpicture}
  \caption{Constructing the lists of colors: we assign red to $v_1$ and $v_2$ (and add the color red to the list of $v_3$ and $v_4$), blue to $v_1$ and $v_3$, green to $v_1$ and $v_4$, orange to $v_2$ and $v_3$, violet to $v_2$ and $v_4$ and brown to $v_3$ and $v_4$.}
  \label{weak_lists}
\end{figure}
\vspace{-20pt}
\begin{remark}
In the above proof we can reduce the number of used colors by assigning colors only to nonadjacent pair of vertices; we can also reduce the lengths of the lists by adding this color only to a minimum vertex cut for these two nodes.
\end{remark}
\vspace{-5pt}
\section{Conclusion}
\label{conlusion}
\vspace{-3pt}
In this paper, we presented different notions to model various scenarios of shared vulnerabilities in complex networks by assigning colors to the edges or vertices using the framework of color-avoiding percolation developed by Krause et al.~\cite{krause2017color}. We also analyzed the complexity of finding the color-avoiding connected components.
Despite the similarity of the presented
concepts, the associated percolation problems -- seemingly paradoxically -- differ significantly regarding computational complexity.
We showed that the color-avoiding edge-connected components can be found in polynomial time. However, the complexity of finding the color-avoiding vertex-connected components highly depends on the exact definition, using a strong version the problem is NP-hard, while using a weaker notion makes it possible to find the components in polynomial time.

\vspace{-2pt}
\subsubsection{Acknowledgment.} We thank Michael Danziger, Panna Fekete and Bal\'azs R\'ath for useful conversations.
The research reported in this paper was supported by the BME- Artificial Intelligence FIKP grant of EMMI (BME FIKP-MI/SC).
The publication is also supported by the EFOP-3.6.2-16-2017-00015 project entitled "Deepening the activities of  HU-MATHS-IN, the Hungarian Service Network for Mathematics in Industry and Innovations" through University of Debrecen. The work of both authors is partially supported by the NKFI FK 123962 grant. R. M. is supported by NKFIH K123782 grant and by MTA-BME Stochastics Research Group.

\bibliographystyle{abbrv}
\bibliography{coloravoiding}

\end{document}